\pdfoutput=1
\RequirePackage{ifpdf}
\ifpdf 
\documentclass[pdftex]{sigma}
\else
\documentclass{sigma}
\fi

\usepackage{mathtools}

\usepackage{mathrsfs}

\numberwithin{equation}{section}

\newtheorem{Th}{Theorem}[section]
\newtheorem*{Theorem*}{Theorem}
\newtheorem{Cor}[Th]{Corollary}
\newtheorem{Lem}[Th]{Lemma}

 { \theoremstyle{definition}
\newtheorem{Def}[Th]{Definition}

\newtheorem{Ex}[Th]{Example}
 }

\newcommand{\diff}[2]{\frac{\partial #1}{\partial #2}}

\newcommand{\qa}{\alpha}
\newcommand{\qb}{\beta}
\newcommand{\qd}{\delta}
\newcommand{\qg}{\gamma}

\newcommand{\qt}{\tau}

\newcommand{\qe}{\varepsilon}

\newcommand{\qp}{\partial}

\newcommand{\Ql}{\Lambda}

\newcommand{\Qo}{\Omega}
\newcommand{\mgn}{\overline{\mathcal M}_{g,n}}
\newcommand{\res}{\operatorname{res}}
\newcommand{\vard}[2]{\frac{\delta #1}{\delta #2}}

\newcommand{\kk}[1]{\left(#1\right)}
\newcommand{\fk}[2]{\left[#1, #2\right]}
\newcommand{\rr}{\mathcal R}
\newcommand{\ol}{\overline L}
\begin{document}

\allowdisplaybreaks

\newcommand{\arXivNumber}{2110.03317}

\renewcommand{\PaperNumber}{037}

\FirstPageHeading

\ShortArticleName{Reduction of the 2D Toda Hierarchy and Linear Hodge Integrals}

\ArticleName{Reduction of the 2D Toda Hierarchy\\ and Linear Hodge Integrals}

\Author{Si-Qi LIU, Zhe WANG and Youjin ZHANG}

\AuthorNameForHeading{S.-Q.~Liu, Z.~Wang and Y.~Zhang}
\Address{Department of Mathematical Sciences, Tsinghua University, Beijing 100084, P.R.~China}
\Email{\href{mailto:liusq@tsinghua.edu.cn}{liusq@tsinghua.edu.cn}, \href{mailto:zhe-wang17@mails.tsinghua.edu.cn}{zhe-wang17@mails.tsinghua.edu.cn}, \href{mailto:youjin@tsinghua.edu.cn}{youjin@tsinghua.edu.cn}}

\ArticleDates{Received October 28, 2021, in final form May 15, 2022; Published online May 18, 2022}

\Abstract{We construct a certain reduction of the 2D Toda hierarchy and obtain a tau-symmetric Hamiltonian integrable hierarchy. This reduced integrable hierarchy controls the linear Hodge integrals in the way that one part of its flows yields the intermediate long wave hierarchy, and the remaining flows coincide with a certain limit of the flows of the fractional Volterra hierarchy which controls the special cubic Hodge integrals. }


\Classification{53D45; 37K10; 37K25}

\section{Introduction}
Let $\mgn$ be the moduli space of stable curves of genus $g$ with $n$ marked points, and $\mathbb L_i$ be the $i$-th tautological line bundle of $\mgn$ whose first Chern class is denoted by $\psi_i$ for $i=1,\dots,n$. Let $\mathbb E_{g,n}$ be the Hodge bundle of $\mgn$ and $\qg_j\in H^{j}(\mgn)$ be the degree $j$ component of its Chern character. In \cite{dubrovin2016hodge}, the following generating function of Hodge integrals is studied:
\begin{equation}\label{Hodge}
\mathcal H(\textbf{t};\textbf{s};\qe)=\sum_{\substack{g,n,m\geq 0\\k_1,\dots,k_n\geq 0 \\ l_1,\dots,l_m\geq 1}}\qe^{2g-2}\frac{t_{k_1}\cdots t_{k_n}}{n!}\frac{s_{l_1}\cdots s_{l_m}}{m!}\int_{\mgn}\psi_1^{k_1}\cdots \psi_n^{k_n}\qg_{2l_1-1}\cdots\qg_{2l_m-1}.
\end{equation}
Note that only odd components of the Chern character are considered due to the vanishing of even components $\qg_{2j}$ by Mumford's relation \cite{mumford1983towards}. It is proved that the evolutions of the two-point function \[w=\qe^2\frac{\qp^2}{\qp t_0^2}\mathcal H(\textbf{t};\textbf{s};\qe)\]
along the time variables $t_n$ form a tau-symmetric Hamiltonian integrable hierarchy which is called the Hodge hierarchy.

When the parameters $s_k$ are taken to be equal to some special values, the Hodge hierarchy degenerates to some well-known integrable hierarchies. For example, by taking $s_k=0$, we recover the Korteweg--de Vries (KdV) hierarchy. When $s_k=(2k-2)!s^{2k-1}$, \eqref{Hodge} reduces to a generating function of linear Hodge integrals and the corresponding integrable hierarchy is proved to be the intermediate long wave (ILW) hierarchy in \cite{buryak2015dubrovin}. For arbitrary given non-zero numbers $p$, $q$, $r$ satisfying the local Calabi--Yau condition\[\frac1p+\frac1q+\frac1r = 0,\] we obtain the generating function of the special cubic Hodge integrals
\begin{equation}\label{cubicHodge}
\mathcal H(\textbf{t};p,q,r;\qe)=\sum_{g\geq 0}\qe^{2g-2}\sum_{n\geq 0}\frac{t_{k_1}\cdots t_{k_n}}{n!}\int_{\mgn}\psi_1^{k_1}\cdots \psi_n^{k_n}\mathcal C_g(-p)\mathcal C_g(-q)\mathcal C_g(-r)
\end{equation}
from \eqref{Hodge} by setting
\[
s_j=-(2j-2)!\big(p^{2j-1}+q^{2j-1}+r^{2j-1}\big),
\]
where $\mathcal C_g(z)$ is the Chern polynomial of the Hodge bundle $\mathbb E_{g,n}$. In \cite{liu2021hodge}, it is proved that the corresponding integrable hierarchy is the fractional Volterra (FV) hierarchy which is constructed in \cite{liu2018fractional}. This fact is called the Hodge-FVH correspondence.

For a fixed parameter $p\neq 0$, we see from the local Calabi--Yau condition $pq+qr+rp=0$ that when $q$ tends to zero so does $r = -pq/(p+q)$. Hence by taking such a limit, the generating function \eqref{cubicHodge} becomes a generating function of linear Hodge integrals.
A natural question is whether we can also take the limit of the FV hierarchy to obtain the ILW hierarchy? The answer to this question is not straightforward due to the fact that the construction of the FV hierarchy involves a complicated infinite linear combination of the time variables \cite{liu2021hodge}. It turns out that one family of the flows of the FV hierarchy does not admit a limit when $q,r\to 0$ (see Section~\ref{sec3} for details) and the other family does have a limit, and this limit can be viewed as infinite linear combinations of the flows of the ILW hierarchy.

To illustrate the limit procedure more explicitly, let us consider the following equation which is one of the flows of the FV hierarchy:
\[
\diff{u}{t} = \frac{\big(\Ql^{-1/r}-1\big)\big(1-\Ql^{-1/q}\big)}{\qe\big(\Ql^{1/p}-1\big)}{\rm e}^u,
\]
where $\Ql = \exp(\qe\qp_x)$ is the shift operator. After the rescaling $\qe\mapsto q\qe$, $t\mapsto q^2t$, we can take the limit $q\to 0$ of the above equation, and obtain the following equation:
\begin{equation}\label{EA}
\diff{u}{t}=p\frac{\big(1-\Ql^{-1}\big)(\Ql-1)}{\qe^2\qp_x}{\rm e}^u = p{\rm e}^uu_x+p\frac{\qe^2}{12}{\rm e}^u\big(u_{xxx}+3u_xu_{xx}+u_x^3\big)+O\big(\qe^4\big).
\end{equation}
On the other hand, from the Lax equations of the ILW hierarchy \cite{buryak2018simple} it follows that the first nontrivial flow of the ILW hierarchy reads
\[
\diff{w}{s_1} = ww_x+\qe\frac{\qt\qp_x^2}{2}\frac{\Ql+1}{\Ql-1}w-\qt w_x,
\]
and all other flows have the form
\begin{gather*}
\diff{w}{s_n}= \frac{w^n}{n!}w_x+\qe^2\qt\big(a_nw^{n-1}w_{xxx}+b_nw^{n-2}w_{xx}w_x+c_nw^{n-3}w_{x}^3\big)+O\big(\qe^4\big),
\\
a_n = \frac{(n-1)!}{12},\qquad
b_n = \frac{(n-2)!}{6},\qquad
c_n=\frac{(n-3)!}{24},\qquad
n\geq 1,
\end{gather*}
here $\qt$ is a parameter of the ILW hierarchy and we assume $n! = 0$ for $n<0$. Let us define a~flow by an infinite linear combination of the flows of the ILW hierarchy:
\[
\diff{}{s}:=\sum_{n\geq 1}\diff{}{s_n},
\]
then we obtain the following expression for this flow
\[
\diff{w}{s} = {\rm e}^ww_x+\frac{\qe^2}{24}{\rm e}^w\big(2w_{xxx}+4w_xw_{xx}+w_x^3\big)+O\big(\qe^4\big),
\]
here we take $\qt=1$ for simplicity.
It is straightforward to verify that after a Miura type transformation
\[
u = w+\frac{\qe^2}{24}w_{xx}+O\big(\qe^4\big),
\]
we arrive at
\[
p\diff{u}{s} = p{\rm e}^uu_x+p\frac{\qe^2}{12}{\rm e}^u\big(u_{xxx}+3u_xu_{xx}+u_x^3\big)+O\big(\qe^4\big),
\]
which coincides with the flow \eqref{EA}. Therefore we see that the limit of the FV hierarchy is related to the ILW hierarchy by infinite linear combinations of flows.

Since both the ILW hierarchy and the limit of the FV hierarchy correspond to the linear Hodge integrals, we expect that these two integrable hierarchies are compatible with each other and that there exists an integrable hierarchy that contains these two hierarchies. It turns out that such an integrable hierarchy indeed exists and it is given by a certain reduction of the 2D Toda hierarchy.

The 2D Toda hierarchy \cite{toda1967vibration,ueno1984toda} is one of the central objects of study in the theory of integrable systems and its various reductions play important roles in the field of mathematical physics. For example, the 1D Toda hierarchy (and its extension), the equivariant Toda hierarchy and the Ablowitz--Ladik hierarchy control the Gromov--Witten theory of~$\mathbb P^1$~\cite{carlet2004extended,getzler2001toda}, the equivariant Gromov--Witten theory of $\mathbb P^1$ \cite{okounkov2006equivariant} and the Gromov--Witten theory of local $\mathbb P^1$ \cite{brini2012local} respectively. For more applications of the 2D Toda hierarchy, one may refer to \cite{takasaki2018toda} and the references therein. Recall that the 2D Toda hierarchy can be described in terms of the operators \cite{ueno1984toda}
\[
L = \Ql+ \sum_{n\geq 0}u_n \Ql^{-n},\qquad
\ol = \bar u_{-1}\Ql^{-1}+\sum_{n\geq 0}\bar u_n\Ql^n
\]
by the following Lax equations:
\begin{gather*}
\diff{L}{t_{1,n}} = \big[(L^n)_+,L\big],\qquad \diff{L}{t_{2,n}} = -\big[\big(\ol^n\big)_-,L\big],
\\
\diff{\ol}{t_{1,n}} = \big[(L^n)_+,\ol\,\big],\qquad \diff{\ol}{t_{2,n}} = -\big[\big(\ol^n\big)_-,\ol\,\big].
\end{gather*}
Let us consider the following reduction of the 2D Toda hierarchy
\[
\log L = p\ol-\log\ol = p K,
\]
where the precise definition of the logarithm of $L$ and $\ol$ will be given in Section~\ref{sec2},
\begin{equation}\label{BN}
K=\frac1{p}\qe\qp_x+{\rm e}^u\Ql^{-1},
\end{equation}
and $p$ is a parameter. Then we define the reduction of the flows by
\[
\qe\diff{K}{t_{1,n}} = [(L^n)_+,K],\qquad
\qe\diff{K}{t_{2,n}} = -\big[\big(\ol^n\big)_-,K\big],\qquad
n\geq 1.
\]
This is a well-defined reduction of the 2D Toda hierarchy, and due to the reason that the construction of this reduction is inspired by taking the limit of the FV hierarchy, we call this integrable hierarchy the limit fractional Volterra (LFV) hierarcy. We summarize its main properties in the following theorem.
\begin{Th}
The LFV hierarchy is a tau-symmetric Hamiltonian integrable hierarchy with hydrodynamic limit. Moreover, the flows $\diff{}{t_{1,n}}$ of the LFV hierarchy are certain limits of the flows of the FV hierarchy and the flows $\diff{}{t_{2,n}}$ is equivalent to the ILW hierarchy under a certain Miura-type transformation.
\end{Th}

The paper is organized as follows. In Section~\ref{sec2}, we give the definition of the LFV hierarchy and prove that it is a tau-symmetric Hamiltonian integrable hierarchy. In Section~\ref{sec3}, we explain how the construction of the LFV hierarchy is inspired by taking a certain limit of the FV hierarchy, and we also obtain a limit of the Hodge-FVH correspondence. In Section~\ref{sec4}, we relate the LFV hierarchy to the ILW hierarchy. Finally in Section~\ref{sec5}, we give some concluding remarks about the relation between our work and the Gromov--Witten/Hurwitz theory.

\section{The LFV hierarchy and its properties}\label{sec2}
Throughout this paper, we work with the ring of differential polynomials $\mathcal R(u)$. It consists of formal power series in $\qe$ with coefficients being elements in the polynomial ring $C^\infty(u)\otimes\mathbb C\big[u^{(k)}\colon k\geq 1\big]$. Let us define a derivation $\qp_x$ and an automorphism $\Ql$ on $\rr(u)$ by
\[
\qp_x = \sum_{k\geq 0}u^{(k+1)}\diff{}{u^{(k)}},\qquad \Ql = \exp(\qe\qp_x).
\]
If we view $u = u(x)$ as a function of the spatial variable $x$, then it is easy to see that $u^{(k)} = \qp_x^ku(x)$ and $\Ql u(x) = u(x+\qe)$. For this reason, the operator $\Ql$ is called the shift operator. Note that the ring $\rr(u)$ is graded with respect to the differential degree $\deg_x$ given by $\deg_x u^{(k)} = k$.

Let us consider the following two difference operators given by
\begin{gather}
\label{BG}
L=\Ql+a_0+a_1\Ql^{-1}+\cdots,
\\
\label{BH}
\overline L={\rm e}^u\Ql^{-1}+b_0+b_1\Ql+\cdots,
\end{gather}
where $a_i,b_i\in\rr(u)$ for $i\geq 0$. In \cite{ueno1984toda}, the dressing operators $P$, $Q$ for these difference operators are defined by
\begin{alignat}{3}
\label{BQ}
& L=P\Ql P^{-1}, \qquad&& P=1+\sum_{k\geq 1}p_k\Ql^{-k},&
\\
\label{BR}
& \overline L = Q\Ql^{-1} Q^{-1},\qquad && Q=\sum_{k\geq 0}q_k\Ql^k.&
\end{alignat}
The coefficients $p_k$ and $q_k$ of the dressing operators $P$ and $Q$ are not in the ring $\mathcal R(u)$ but only exist in a certain extension of $\mathcal R(u)$ (for details, see \cite{ueno1984toda}). Note that the choice of $P$ and~$Q$ is unique up to the right multiplication by difference operators with constant coefficients. Follo\-wing~\cite{carlet2004extended}, we define the logarithm of the operators $L$, $\overline{L}$ as follows:
\begin{gather*}
\log L:=P(\qe\qp_x)P^{-1}=\qe\qp_x-\qe P_xP^{-1},
\\
\log\overline L:=Q(-\qe\qp_x)Q^{-1}=-\qe\qp_x+\qe Q_xQ^{-1},
\end{gather*}
where $P_x = \sum_{k\geq 1}(\qp_xp_k)\Ql^{-k}$ and $Q_x = \sum_{k\geq 0}(\qp_xq_k)\Ql^k$. The ambiguities of choices of $P$ and~$Q$ are canceled in the operators $P_xP^{-1}$ and $Q_xQ^{-1}$, and they are difference operators with coefficients belonging to $\mathcal R(u)$ \cite{carlet2004extended}. Before giving the definition of the LFV hierarchy, let us make some preparations by proving the following lemmas.
\begin{Lem}
\label{BI}
There exist unique differential polynomials $a_k$ such that
\begin{equation}\label{zh-2022-2}
\lim_{\qe\to 0}a_k = \frac{p^{k+1}}{(k+1)!}{\rm e}^{(k+1)u}
\end{equation}
and the difference operator $L$ defined in \eqref{BG} satisfies the relation
\begin{equation}\label{zh-2022-1}
\frac{1}{p}\log L = K,
\end{equation}
here $p$ is a formal parameter and $K$ is the differential-difference operator defined by \eqref{BN}.
\end{Lem}
\begin{proof}
Let us first find $a_i\in\rr(u)$ such that the operator $L$ defined in \eqref{BG} satisfies the identities
\begin{gather}
\label{BJ}
\res \big[L^n,p{\rm e}^u\Ql^{-1}\big] = \qe\qp_x\res L^n,\qquad n\geq 1,
\end{gather}
here and henceforth, for any difference operator $D = \sum_k f_k\Ql^k$ with $f_k\in\rr(u)$, we define its residue by $\res D = f_0$. By taking $n=1$ in the equation \eqref{BJ}, we arrive at
\[
p(\Ql-1){\rm e}^u = \qe\qp_xa_0.
\]
Then the above equation for $a_0$ has a unique solution by taking the integral constant to be zero, i.e.,
\begin{equation}
\label{BL}
a_0 = p\frac{\Ql-1}{\qe\qp_x}{\rm e}^u.
\end{equation}

For general $n\geq 1$, one can show by induction that if we represent $L^n = \sum_kf^n_k\Ql^{k}$, then the differential polynomials $f^n_k$ can be viewed as functions in $a_i$ and we have
\begin{gather*}
f^n_k = f^n_k(a_0,\dots,a_{n-1-k}),\qquad k\leq n-1,
\\
f^n_0 = \big(1+\Ql+\cdots+\Ql^{n-1}\big)a_{n-1}+g_n(a_0,\dots,a_{n-2}).
\end{gather*}
Therefore we see that the differential polynomials can be found recursively by taking $n=2,3,\dots$ in the equation \eqref{BJ}. More explicitly, if we have found differential polynomials $a_0,\dots,a_{n-1}$, to determine $a_{n}$, we consider the equation
\[
\res \big[L^{n+1},p{\rm e}^u\Ql^{-1}\big] = \qe\qp_x\res L^{n+1}
\]
and obtain that
\begin{gather*}
p\big(1-\Ql^{-1}\big)\kk{f^{n+1}_1(a_0,\dots,a_{n-1})\Ql {\rm e}^u} =\kk{1+\Ql+\cdots+\Ql^n}\qe\qp_xa_n+\qe\qp_xg_{n+1}(a_0,\dots,a_{n-1}).
\end{gather*}
Hence we can solve $a_n$ uniquely by taking the integral constant to be zero.

Next we show that the difference operator $L$ determined by the equations \eqref{BJ} satisfies the relation \eqref{zh-2022-1}.
Indeed, by using the results given in \cite{carlet2004extended}, we see that the equations \eqref{BJ} imply that
\[
\qe P_xP^{-1} = -p{\rm e}^u\Ql^{-1},
\]
and therefore the relation $\frac 1p\log L = K$ holds true.

Finally we show that the differential polynomials $a_i$ determined above satisfy
the rela\-tion~\eqref{zh-2022-2}.
From the relation \eqref{zh-2022-1} it follows that $\fk{L}{pK} = 0$, which is equivalent to the fol\-lo\-wing recursion relation:
\begin{equation}
\label{BT}
\frac{1}{p}\qe\qp_x a_{k+1}= a_k\Ql^{-k}{\rm e}^{ u}-{\rm e}^{ u}\Ql^{-1} a_k,\qquad k\geq 0.
\end{equation}
Then \eqref{zh-2022-2} can be verified using this recursion relation and the initial condition \eqref{BL}. The lemma is proved.
\end{proof}

\begin{Lem}
\label{BM}
There exist unique differential polynomials $b_k$ such that
\[
\lim_{\qe\to 0}b_k = \frac{{\rm e}^{-ku}}{p^{k+1}}\qb_k
\]
and the difference operator $\ol$ defined in \eqref{BH} satisfies the relations
\[
\overline L-\frac 1p\log\overline L = K,
\]
where $K$ is the differential-difference operator given by \eqref{BN} and $\qb_k$ are polynomials in $u$ satisfying the recursion relations
\[
\qb_{k+1} = \qb_k-\int_0^uk\qb_k\,{\rm d}u,\qquad k\geq 0,
\]
with the initial condition $\qb_0 = u$.
\end{Lem}
\begin{proof}
The lemma can be proved by using a similar method that is used in the proof of Lemma~\ref{BI}, so we omit the details here. For later use, we write down the following recursion relations satisfied by $b_k$:
\begin{equation}
\label{BU}
\frac{1}{p}\qe\qp_x b_{k}= b_{k+1}\Ql^{k+1}{\rm e}^{u}-{\rm e}^{u}\Ql^{-1}b_{k+1},\qquad k\geq -1,
\end{equation}
with $b_{-1}={\rm e}^{ u}$.
\end{proof}
\begin{Def}
The limit fractional Volterra (LFV) hierarchy consists of the flows
\begin{equation}
\label{def-LFVH}
\qe\diff{K}{t_{1,n}} = [(L^n)_+,K],\qquad \qe\diff{K}{t_{2,n}} = -\big[\big(\ol^n\big)_-,K\big],\qquad n\geq 1,
\end{equation}
where the operator $K$ is given by \eqref{BN}, and the opertors $L$ and $\ol$ are determined by Lemma~\ref{BI} and Lemma \ref{BM} respectively. Here and in what follows, for a difference operator $\sum_k f_k\Ql^k$, we define its positive part and negative part by
\[
\biggl(\sum_k f_k\Ql^k\biggr)_+ = \sum_{k\geq 0}f_k\Ql^k,\qquad
\biggl(\sum_k f_k\Ql^k\biggr)_- = \sum_{k< 0}f_k\Ql^k.
\]
\end{Def}

It follows from Lemmas \ref{BI} and~\ref{BM} that the operators $L$ and $\ol$ commute with $K$, therefore we have the following identities:
\begin{equation*}
[(L^n)_+,K] = -[(L^n)_-,K],\qquad
\big[\big(\ol^n\big)_+,K\big]=-\big[\big(\ol^n\big)_-,K\big].
\end{equation*}
Since $[(L^n)_+,K]$ is a difference operator of the form $\sum_{j\geq -1}f_j\Ql^j$ and $[(L^n)_-,K]$ is of the form $\sum_{j\leq -1}g_j\Ql^j$, it follows that the first set of equations given in \eqref{def-LFVH} yields a hierarchy of well-defined equations of $u$. Similarly the second set of equations given in \eqref{def-LFVH} is also well-defined. Moreover, it is easy to see that the flows of the LFV hierarchy are given by differential polynomials with hydrodynamic limits.
\begin{Ex}
The simplest flows of \eqref{def-LFVH} read
\begin{gather*}
\diff{u}{t_{1,1}}=p\frac{\big(1-\Ql^{-1}\big)(\Ql-1)}{\qe^2\qp_x}{\rm e}^u,\\
\diff{u}{t_{2,1}}=\frac 1p u_x,\qquad \diff{u}{t_{2,2}}=\frac{u_x}{p^2}\qe\qp_x\frac{\Ql+1}{\Ql-1}u+\frac{1}{p^2}\qe\qp_x^2\frac{\Ql+1}{\Ql-1}u.
\end{gather*}
\end{Ex}

Let us proceed to present some basic properties of the LFV hierarchy. We will show that it is a tau-symmetric Hamiltonian integrable hierarchy. We start with proving that the flows defined in \eqref{def-LFVH} mutually commute. The following lemma is standard in the theory of integrable hierarchies.
\begin{Lem}\label{eq-L}
The following equations are satisfied by the operators $L$ and $\overline L$:
\begin{alignat*}{3}
& \qe\diff{L}{t_{1,n}}= [(L^n)_+,L],\qquad&&
\qe\diff{L}{t_{2,n}} = -\big[\big(\ol^n\big)_-,L\big],&
\\
& \qe\diff{\overline L}{t_{1,n}} = \big[(L^n)_+,\overline L\,\big],\qquad&&
\qe\diff{\overline L}{t_{2,n}} = -[\big(\ol^n\big)_-,\overline L\,\big].&
\end{alignat*}
\end{Lem}
From Lemma \ref{eq-L}, it is straightforward to derive the commutation relations of the flows~\eqref{def-LFVH}.
\begin{Th}
The flows \eqref{def-LFVH} of the LFV hierarchy mutually commute, i.e.,
\[
\bigg[\diff{}{t_{1,n}},\diff{}{t_{1,m}}\bigg]=\bigg[\diff{}{t_{2,n}},\diff{}{t_{2,m}}\bigg] =\bigg[\diff{}{t_{1,n}},\diff{}{t_{2,m}}\bigg]=0,\qquad n,m\geq 1.
\]
\end{Th}

To derive the Hamiltonian formalism of the LFV hierarchy, we need to compute the variational derivatives of the local functionals of the form
\[
\int\res L^n,\qquad \int\res\overline L^n.
\]
Let us give a general description on how such a variational derivative can be computed (see~\cite{carlet2004extended}).

Consider the 1-form
\[
\sum_{k\geq 0}f_k{\rm d}u^{(k)},\qquad
f_k\in\mathcal R(u),
\]
where ${\rm d}$ is the natural exterior differential operator on $\mathcal R(u)$. For two 1-forms $\sum_{k\geq 0}f_k{\rm d}u^{(k)}$ and $\sum_{k\geq 0}g_k{\rm d}u^{(k)}$, we denote $\sum_{k\geq 0}f_k{\rm d}u^{(k)}\sim\sum_{k\geq 0}g_k{\rm d}u^{(k)}$ if there exists another 1-form $\sum_{k\geq 0}h_k{\rm d}u^{(k)}$ such that:
\[
\sum_{k\geq 0}f_k{\rm d}u^{(k)}-\sum_{k\geq 0}g_k{\rm d}u^{(k)} = \qp_x\biggl(\sum_{k\geq 0}h_k{\rm d}u^{(k)}\biggr),
\]
here the derivation $\qp_x$ acts on the 1-form by $\qp_x{\rm d}u^{(k)} = {\rm d}u^{(k+1)}$. Now for a local functional $\int h$, $h\in\mathcal R(u)$, we can compute its variational derivative as follows:
\begin{equation}\label{form-vard}
{\rm d}h = \sum_{k\geq 0}\diff{h}{u^{(k)}}{\rm d}u^{(k)}\sim\bigg(\vard{}{u}\int h\bigg){\rm d}u.
\end{equation}

\begin{Lem}
\label{temp2}
Consider the local functionals defined by
\begin{equation*}
H_n = \frac{1}{np}\int\res L^n,\qquad n\geq 1.
\end{equation*}
Then their variational derivatives have the expressions
\begin{equation}\label{var-hn}
\vard{H_n}{u} = \frac 1p\frac{\qe\qp_x}{\Ql-1}\res L^n,\qquad
n\geq 1.
\end{equation}
\end{Lem}
\begin{proof}
We need the following identity whose proof can be found in \cite{carlet2004extended}:
\begin{equation}\label{temp3}
\res L^{n-1}{\rm d}L\sim-\res L^n{\rm d}\big(\qe P_xP^{-1}\big),
\end{equation}
where the operator $P$ is defined in \eqref{BQ}. From this identity and the relation \eqref{zh-2022-1} it follows that
\[
{\rm d}\bigg(\frac{1}{np}\res L^n\bigg)\sim\frac{1}{p}\res L^{n-1}{\rm d}L\sim-\frac{1}{p}\res L^n{\rm d}\big(\qe P_xP^{-1}\big) = \res L^n\big({\rm e}^u{\rm d}u\Ql^{-1}\big).
\]
Here and henceforth, by abusing the notations we denote
\[
{\rm d}\bigg(\sum_{n}f_k\Ql^k\bigg)= \sum_{k}({\rm d}f_k)\Ql^k
\]
for a difference operator $\sum_{k}f_k\Ql^k$.

Let us represent $L^n$ in the form $L^n= \sum_k a_{k,n}\Ql^k$, we then arrive at the following relations:
\[
{\rm d}\bigg(\frac{1}{np}\res L^n\bigg)\sim a_{1,n}\Ql\big({\rm e}^u{\rm d}u\big)\sim \big({\rm e}^u\Ql^{-1}a_{1,n}\big){\rm d}u.
\]
Hence by using the formula \eqref{form-vard} we obtain
\[
\vard{H_n}{u} = {\rm e}^u\Ql^{-1}a_{1,n},\qquad n\geq 1.
\]
Finally by using the identity $\res[L^n,p{\rm e}^u\Ql^{-1}] = \res\qe\qp_xL^n$ obtained from Lemma \ref{BI}, we prove the desired identity \eqref{var-hn}. The lemma is proved.
\end{proof}

\begin{Lem}
\label{temp4}
The variational derivatives of the local functionals
\begin{equation*}
\overline H_n = \int\bigg(\!\res \frac{\overline L^{n+1}}{n+1}-\res \frac{\overline L^n}{np}\bigg),\qquad n\geq 1
\end{equation*}
can be represented as
\begin{equation}\label{var-bhn}
\vard{\overline H_n}{u} = \frac 1p\frac{\qe\qp_x}{\Ql-1}\res \overline L^n,\qquad n\geq 1.
\end{equation}
\end{Lem}
\begin{proof}
Similar to the identity\eqref{temp3} we have
\[
\res \overline L^{n-1}{\rm d}\overline L\sim\res \overline L^n{\rm d}\big(\qe Q_xQ^{-1}\big).
\]
By using Lemma \ref{BM} we obtain the relations
\begin{align*}
\res \overline L^n{\rm d}K &= \res \overline L^n{\rm d}\overline L-\frac 1p\res \overline L^n{\rm d}\big(\qe Q_xQ^{-1}\big)
\\
&\sim \res \overline L^n{\rm d}\overline L-\frac 1p \res\overline L^{n-1}{\rm d}\overline L
\\
&\sim \res {\rm d}\bigg(\frac{\overline L^{n+1}}{n+1}- \frac{\overline L^n}{np}\bigg).
\end{align*}
We then arrive at \eqref{var-bhn} by applying the calculation similar to the one we do in the proof of Lemma \ref{temp2}. The lemma is proved.
\end{proof}

\begin{Th}
The flows \eqref{def-LFVH} can be represented as Hamiltonian systems as follows:
\[
\qe\diff{u}{t_{1,n}} = \{u(x), H_n\},\qquad
\qe\diff{u}{t_{2,n}} = \big\{u(x),\overline H_n\big\},\qquad n\geq 1,
\]
where the Poisson bracket $\{-,-\}$ is defined by the Hamiltonian operator
\begin{equation*}
\mathscr{P} = p\frac{\big(1-\Ql^{-1}\big)(\Ql-1)}{\qe^2\qp_x}.
\end{equation*}
\end{Th}
\begin{proof}
From the definition \eqref{def-LFVH} it is straightforward to see that
\[
\qe\diff{u}{t_{1,n}} = \big(1-\Ql^{-1}\big)\res L^n =\qe \mathscr{P}\vard{H_n}{u} = \qe\{u(x),H_n\}.
\]
As for the flows $\diff{u}{t_{2,n}}$, let us denote $\overline L^n=\sum_{k}b_{k,n}\Ql^k$, then by combining the definition \eqref{def-LFVH} and the definitions of $L$ and $\ol$, we arrive at
\[
\qe {\rm e}^u\diff{u}{t_{2,n}} = \frac 1p \qe\qp_x b_{-1,n} = {\rm e}^u\big(1-\Ql^{-1}\big)\res \overline L^n,
\]
which implies the Hamiltonian formalism of the flows $\diff{u}{t_{2,n}}$ due to Lemma~\ref{temp4}. Finally it is obvious that $\mathscr P$ is indeed a Hamiltonian operator and thus the theorem is proved.
\end{proof}

Finally we are going to consider the tau-structure of the LFV hierarchy. The constructions and proofs are standard and very similar to those presented in~\cite{liu2018fractional}.
\begin{Lem}\label{temp5}
We define the following functions for $k,l\geq 1{:}$
\begin{gather}\label{o11}
\Qo_{1,k;1,l}= \sum_{n=1}^l\frac{\Ql^n-1}{\Ql-1}\big(\!\res\big(\Ql^{-n}L^l\big)\res\big(L^k\Ql^n\big)\big),
\\
\label{o21}
\Qo_{2,k;1,l} =\Qo_{1,l;2,k} =\sum_{n=1}^l\frac{\Ql^n-1}{\Ql-1}\big(\!\res\big(\Ql^{-n}L^l\big)\res\big(\overline L^k\Ql^n\big)\big),
\\
\label{o22}
\Qo_{2,k;2,l} =\sum_{n=1}^l\frac{\Ql^n-1}{\Ql-1}\big(\!\res\big(\Ql^{-n}\overline L^l\big)\res\big(\overline L^k\Ql^n\big)\big).
\end{gather}
Then the following relations hold true:
\begin{gather*}
\qe\diff{}{t_{1,l}}\res L^k = \qe\diff{}{t_{1,k}}\res L^l = (\Ql-1)\Qo_{1,k;1,l},
\\
\qe\diff{}{t_{1,l}}\res \overline L^k = \qe\diff{}{t_{2,k}}\res L^l = (\Ql-1)\Qo_{2,k;1,l},
\\
\qe\diff{}{t_{2,l}}\res \overline L^k = \qe\diff{}{t_{2,k}}\overline \res L^l = (\Ql-1)\Qo_{2,k;2,l}.
\end{gather*}
\end{Lem}

\begin{Lem}
\label{temp6}
The functions defined in \eqref{o11}--\eqref{o22} satisfy the following identities:
\[
\Qo_{i,k;j,l} = \Qo_{j,l;i,k},\qquad
\diff{\Qo_{j,l;m,n}}{t_{i,k}} = \diff{\Qo_{i,k;m,n}}{t_{j,l}},\qquad
i,j,m=1,2,\quad k,l,n\geq 1.
\]
\end{Lem}
\begin{proof}
Let us start by proving the first set of identities. It follows from Lemma \ref{temp5} that
\[
(\Ql-1)\Qo_{i,k;j,l} = (\Ql-1)\Qo_{j,l;i,k}.
\]
Therefore there exist constants $c_{i,k;j,l}$ such that
\[
\Qo_{i,k;j,l}-\Qo_{j,l;i,k}=c_{i,k;j,l}.
\]
To verify that the constants $c_{j,k;j,l}$ vanish, we may compute the limits of the differential polynomials $\Qo_{i,k;j,l}$ by setting $u = 0$ and $u^{(k)} = 0$. It is easy to see from Lemmas \ref{BI} and~\ref{BM} that the only non-trivial case is $i=j=1$. By using Lemma \ref{BI}, it follows from a straightforward computation that
\[
\Qo_{1,k;1,l}|_{u=u^{(i)} = 0} = p^{k+l}k^kl^l\sum_{n=1}^l \kk{\frac kl}^n\frac{n}{(l-n)!(k+n)!}.
\]
The summation on the right hand side of the above identity is evaluated using Gosper's algorithm~\cite{gosper1978decision}.\footnote{Here we use a Mathematica package provided on the website \url{https://www2.math.upenn.edu/~wilf/progs.html}.} Let us denote
\[
y_n = \kk{\frac kl}^n\frac{n}{(l-n)!(k+n)!},\qquad 1\leq n\leq l,
\]
and
\[
z_n = -\kk{\frac kl}^n\frac{l}{(k+l)(l-n)!(k+n-1)!},\qquad 1\leq n\leq l.
\]
Then it is straightforward to verify that
\[
z_{n+1}-z_n = y_n,\qquad 1\leq n\leq l-1,\qquad z_l = -y_l.
\]
Hence we conclude that
\[
\sum_{n=1}^l \kk{\frac kl}^n\frac{n}{(l-n)!(k+n)!} = -z_1 = \frac{1}{(k+l)(k-1)!(l-1)!}.
\]
So we see that $\Qo_{1,k;1,l}|_{u=u^{(i)} = 0}$ is symmetric with respect to the indices $k$ and $l$ and therefore we see that $\Qo_{1,k;1,l} = \Qo_{1,l;1,k}$.

On the other hand, it follows from Lemma \ref{temp5} that
\begin{equation}
\label{BS}
(\Ql-1)\diff{\Qo_{j,l;m,n}}{t_{i,k}} = (\Ql-1)\diff{\Qo_{i,k;m,n}}{t_{j,l}}.
\end{equation}
Since the differential polynomials $\diff{\Qo_{j,l;m,n}}{t_{i,k}}$ have differential degrees greater than $1$, from \eqref{BS} we arrive at the validity of the second set of identities of the lemma. The lemma is proved.
\end{proof}

\begin{Th}
For any solution $u(x;\textbf{t})$ of the LFV hierarchy, there exists a tau-function $\qt(x;\textbf{t})$ such that:
\begin{gather}\label{temp7}
(\Ql-1)\big(1-\Ql^{-1}\big)\log\qt = u,
\\
\qe(\Ql-1)\diff{\log\qt}{t_{1,k}} = \res L^k,
\\
\qe(\Ql-1)\diff{\log\qt}{t_{2,k}} = \res \overline L^k,
\\
\label{temp8}
\qe^2\frac{\qp^2\log\qt}{\qp t_{i,k}\qp t_{j,l}} = \Qo_{i,k;j,l}.
\end{gather}
\end{Th}
\begin{proof}
The compatibility of the equations \eqref{temp7}--\eqref{temp8} is given by the definition of LFV hie\-rarchy \eqref{def-LFVH}, Lemmas \ref{temp5} and~\ref{temp6}. Hence such a tau-function exists and we prove the theorem.\looseness=1
\end{proof}

\section{The LFV hierarchy as a limit of the FV hierarchy}\label{sec3}
Let us explain how the LFV hierarchy can be viewed as a certain limit of the FV hierarchy given in \cite{liu2018fractional}. Let $p$, $q$, $r$ be any given non-zero complex numbers satisfying the condition $pq+qr+rp=0$. We introduce the shift operators
\[
\Ql_1=\Ql^{1/q},\qquad
\Ql_2= \Ql^{1/p},\qquad
\Ql_3=\Ql_1\Ql_2=\Ql^{-1/r},\qquad
\Ql = {\rm e}^{\qe\qp_x},
\]
and consider the following Lax operator
\[
L^{\rm FV} =\Ql_2+{\rm e}^v\Ql_1^{-1}
\]
with $v=v(x,\qe)$.
The fractional powers
\[
\mathcal A=\big(L^{\rm FV}\big)^{-p/r}=\Ql_3+\sum_{k\geq 0}f_k\Ql_3^{-k},\qquad
\mathcal B=\big(L^{\rm FV}\big)^{-q/r}=g_{-1}\Ql_3^{-1}+\sum_{k\geq 0}g_k\Ql_3^k
\]
are well-defined, and their coefficients $f_k$, $g_k$ belong to the ring $\rr(v)$.
It follows from the relations
\[
\big[\mathcal A,L^{\rm FV}\big]=\big[\mathcal B,L^{\rm FV}\big]=0
\]
that the coefficients $f_n$ and $g_n$ satisfy the following recursion relations:
\begin{gather}
\label{rec-FVH-1}
(\Ql_2-1)f_{k+1}=f_k\Ql_3^{-k}{\rm e}^v-{\rm e}^v\Ql_1^{-1}f_k,\qquad k\geq 0,\qquad f_0=\frac{\Ql_3-1}{\Ql_2-1}{\rm e}^v,
\\
\label{rec-FVH-2}
(\Ql_2-1)g_{k}=g_{k+1}\Ql_3^{k+1}{\rm e}^v-{\rm e}^v\Ql_1^{-1}g_{k+1},\qquad k\geq -1,\qquad
g_{-1}={\rm e}^{\frac{1-\Ql_3^{-1}}{1-\Ql_1^{-1}}v}.
\end{gather}

The fractional Volterra hierarchy is defined by the following Lax equations:
\begin{equation}
\label{def-FVH}
\qe\diff{L^{\rm FV}}{T_{1,n}} = \big[\mathcal A^n_+,L^{\rm FV}\big],\qquad
\qe\diff{L^{\rm FV}}{T_{2,n}} = -\big[\mathcal B^n_-,L^{\rm FV}\big].
\end{equation}
Here for an operator $\mathcal{C}$ of the form $\sum c_k \Lambda_3^k$, we denote
\[
\mathcal{C}_+=\sum_{k\ge 0} c_k \Lambda_3^k,\qquad
\mathcal{C}_-=\sum_{k<0} c_k \Lambda_3^k.
\]
In order to take a certain limit of the FV hierarchy, we first introduce a new dispersion parameter $\tilde\qe=\qe/q$ and the associated shift operator $\tilde \Ql = {\rm e}^{\tilde\qe\qp_x}$. We have the relations
\[
\Ql_1 = \tilde \Ql,\qquad \Ql_3 =\tilde \Ql^{\frac{p+q}{p}},
\]
and we can rewrite the recursion relations \eqref{rec-FVH-1} and \eqref{rec-FVH-2} as follows:
\begin{gather}
\label{rec-FVH-3}
\sum_{i\geq 1}\frac{\tilde \qe^i}{i!}\frac{q^{i-1}}{p^i}\qp_x^i\tilde f_{k+1}=\tilde f_k\tilde \Ql^{-\frac{k(p+q)}{p}}{\rm e}^v-{\rm e}^v\tilde\Ql^{-1}\tilde f_k,\qquad k\geq 0,
\\
\label{rec-FVH-4}
\sum_{i\geq 1}\frac{\tilde \qe^i}{i!}\frac{q^{i-1}}{p^i}\qp_x^i\tilde g_{k}=\tilde g_{k+1}\tilde\Ql^{\frac{(k+1)(p+q)}{p}}{\rm e}^v-{\rm e}^v\tilde\Ql^{-1}\tilde g_{k+1},\qquad k\geq -1,
\end{gather}
here
\[
\tilde f_k = q^{k+1}f_k,\qquad
\tilde g_k = \frac{g_k}{q^{k+1}}.
\]
Then it is easy to see that the recursion relations \eqref{rec-FVH-3} and \eqref{rec-FVH-4} become the relations \eqref{BT} and~\eqref{BU} after taking the limit $q\to 0$.

Let us look at the relations \eqref{rec-FVH-1} and \eqref{rec-FVH-3} more carefully. The coefficients $f_i$ of the operator~$\mathcal A$ can be uniquely determined from \eqref{rec-FVH-1} by requiring that
\[
\lim_{\qe\to 0}f_k = \frac{p^{k+1}{\rm e}^{kv}}{(k+1)!q^{k+1}}\prod_{i=1}^{k+1}\bigg(1+\frac{(i-k)q}{p}\bigg).
\]
Therefore if we assume the limit
\[
u(x,\tilde\qe) = \lim_{q\to 0}v(x,q\tilde\qe)
\]
exists, then it is easy to see from \eqref{rec-FVH-3} that the limits
\begin{equation}
\label{BW}
a_k(x,\qe) = \lim_{q\to 0}q^{k+1}f_k(x,q\tilde\qe)|_{\tilde\qe\mapsto\qe}
=\lim_{q\to 0}\tilde{f}_k(x,q\tilde\qe)|_{\tilde\qe\mapsto\qe},\qquad k\geq 0
\end{equation}
satisfy the relations \eqref{BL} and \eqref{BT}, and are exactly the coefficients of the operator $L$ described in Lemma \ref{BI}.

By using the above-mentioned observation, we can relate the flows $\diff{u}{t_{1,n}}$ of the LFV hierarchy to the flows $\diff{v}{T_{1,n}}$ of the FV hierarchy as follows. From the definition \eqref{def-LFVH} and \eqref{def-FVH}, we can write these flows as follows:
\[
\qe\diff{u}{t_{1,n}} = \big(1-\Ql^{-1}\big)\res L^n,\qquad
\qe\diff{v}{T_{1,n}} = \big(1-\Ql_1^{-1}\big)\res \mathcal A^n.
\]
By using the relations \eqref{BW}, one can prove that
\[
\res L^n = \lim_{q\to 0}\big(\!\res q^n\mathcal A^n|_{\qe\mapsto q\qe}\big).
\]
Thus we arrive at the relation
\[
\diff{u}{t_{1,n}} = \lim_{q\to 0}\bigg(q^{n+1}\diff{v}{T_{1,n}}\bigg|_{\qe\mapsto q\qe}\bigg).
\]
Moreover, the Hamiltonian operator of the FV hierarchy is given by
\[
\mathscr P^{\rm FV} = \frac{\big(1-\Ql_1^{-1}\big)(\Ql_3-1)}{\Ql_2-1},
\]
and it is easy to verify that the Hamiltonian operator of the LFV hierarchy can be obtained by taking the following limit:
\[
\mathscr P = \lim_{q\to 0}\left(q\mathscr P^{\rm FV}|_{\qe\to q\qe}\right).
\]

We note that the flows $\frac{\qp}{\qp t^{2,n}}$ and $\frac{\qp}{\qp T^{2,n}}$ are not related by means of taking the limit $q\to 0$ even though after taking such a limit the recursion relations~\eqref{rec-FVH-4} coincide with the rela\-ti\-ons~\eqref{BU}. Indeed, let us write down the flows
{\samepage\begin{gather*}
\diff{u}{T_{2,1}} = \frac{1}{p}u_x,
\\
\diff{v}{T_{2,1}} =\frac{\Ql_2-\Ql_1^{-1}}{\qe}\exp\bigg(\frac{1-\Ql^{-1}_2}{1-\Ql_1^{-1}}v\bigg) = \frac{p+q}{pq}{\rm e}^vv_x+O(\qe).
\end{gather*}}
We conclude that we cannot obtain the flow $\diff{u}{t_{2,1}}$ by taking the limit of the flow $\diff{v}{T_{2,1}}$.

The fact that the limits of the flows $\diff{v}{T_{2,n}}$ do not exist can also be obtained in view of the Hodge-FVH correspondence \cite{liu2021hodge}. It is proved in \cite{liu2021hodge} that the following function gives a solution of the FVH hierarchy:
\[
v(x;\textbf{T};\qe) = \big(\Ql_1^{1/2}-\Ql_1^{-1/2}\big)\big(\Ql_3^{1/2}-\Ql_3^{-1/2}\big)\mathcal H\bigg(\textbf{t}(x;\textbf{T});p,q,r;\frac{\sqrt{p+q}}{pq}\qe\bigg),
\]
where the variables $t_i$ and $T_{\qa,k}$ are related by
\[
t_i(x;\textbf{T}) = x\qd_{i,0}+\qd_{i,1}-1+\frac{1}{pq}\sum_{k>0}(kp)^{i+1} \binom{k(p+q)/q}{k}T_{1,k}+(kq)^{i+1}\binom{k(p+q)/p}{k}T_{2,k},
\]
and the generating function $\mathcal H\left(\textbf{t}(x;\textbf{T});p,q,r;\qe\right)$ of the special cubic Hodge integral is defined in \eqref{cubicHodge}. Therefore it follows from the relations between $t_i$ and $T_{\qa,k}$ that there does not exist a~way to rescale $T_{2,k}$ by multiplying a suitable power of $q$ such that after taking the limit $q\to 0$ of $t_i$ the variables $T_{2,k}$ are still preserved.

However, by setting $T_{2,k} = 0$ and by performing the change of variables $T_{1,k}\mapsto q^{k+1}t_{1,k}$ and~$\qe\mapsto q\qe$, we can obtain the following corollary which can be viewed as a limit of the Hodge-FVH correspondence.
\begin{Cor}
Let us denote by $\mathcal H(\textbf{t};p;\qe)$ the following generating function of the linear Hodge integrals:
\[
\mathcal H(\textbf{t};p;\qe)=\sum_{g\geq 0}\qe^{2g-2}\sum_{n\geq 0}\frac{t_{k_1}\cdots t_{k_n}}{n!}\int_{\mgn}\psi_1^{k_1}\cdots \psi_n^{k_n}\mathcal C_g(-p)
\]
and denote by $u(x;t_{1,k};\qe)$ the function
\[
u(x;t_{1,k};\qe) = \big(\Ql^{1/2}-\Ql^{-1/2}\big)^2\mathcal H(\textbf{t}(x;t_{1,k});p;\qe),
\]
where the relations between the variables $t_i$ and $t_{1,k}$ are given by
\[
t_i = \sum_{k\geq 0}\frac{k^{i+1+k}}{k!}p^{k+i}t_{1,k}-1+x\qd_{i,0}+\qd_{i,1}.
\]
Then the function $u(x;t_{1,k};\qe)$ satisfies the equations which form a part of the LFV hierarchy:
\[
\qe\diff{K}{t_{1,n}} = [(L^n)_+,K],
\]
where the differential-difference operator $K$ and the difference operator $L$ are given by
\[
K=\frac 1p\qe\qp_x+{\rm e}^u\Ql^{-1},\qquad \log L = pK.
\]
\end{Cor}

\section{Relation between the ILW hierarchy and the LFV hierarchy}\label{sec4}
Due to the discussion given in the last section, we expect that the LFV hierarchy should control a certain generating function of the linear Hodge integrals over the moduli space of stable curves. It is proved in \cite{buryak2015dubrovin} that the integrable hierarchy corresponds to the linear Hodge integral is the intermediate-long wave (ILW) hierarchy. Therefore it is natural to establish relations between the LFV hierarchy and the ILW hierarchy.

We start by reviewing the Lax pair formalism of the ILW hierarchy given in \cite{buryak2018simple}. Compared to the original convention used in \cite{buryak2018simple}, we do some rescalings for the convenience of our presentation given below. Consider the following operators $\mathcal K$ and $\mathcal L$ defined by
\begin{equation}\label{BR-Lax}
\mathcal L-\frac 1p\log\mathcal L = \mathcal K,\qquad
\mathcal K = \Ql+w-\frac 1p\qe\qp_x,
\end{equation}
where the operator $\mathcal L$ can be written as
\[
\mathcal L=\Ql+w+c_1\Ql^{-1}+c_2\Ql^{-2}+\cdots,\qquad
c_k\in\mathcal R(w).
\]
Similar as before, we denote by $\mathcal R(w)$ the ring of differential polynomials of $w$. The logarithm of the operator $\mathcal L$ is defined, similar to $\log L$, via the dressing operator
as follows:
\begin{equation*}
\mathcal L = \mathcal P\Ql \mathcal P^{-1},\qquad
\log\mathcal L = \qe\qp_x-\qe \mathcal P_x\mathcal P^{-1}.
\end{equation*}
Using the operators defined above, the ILW hierarchy can be represented as
\begin{equation}\label{def1-ILW}
\qe\diff{\mathcal K}{s_{1,n}} = [\mathcal L^n_+,\mathcal K],\qquad n\ge 1.
\end{equation}
The differential polynomials $c_k$ are uniquely determined by the recursion
relations	
\begin{equation}
\label{BV}
c_k\big(1-\Ql^{-k}\big)w+(\Ql-1)c_{k+1} = \frac 1p \qe\qp_x c_k,\qquad k\geq 0,\qquad c_0 = w,
\end{equation}
and by the condition that the dispersionless limits of $c_k$ are given by $\lim_{\qe\to 0}c_k = \qg_k$, where $\qg_k$ are polynomials given by
\[
\qg_{k+1} = \frac 1p\qg_k-\int_0^w\qg_k\,{\rm d}w,\qquad k\geq 0,\qquad \qg_0 = w.
\]

The following theorem gives a direct relation between the ILW hierarchy and the LFV hierarchy.
\begin{Th}\label{temp15}
The operator $\overline L$ defined in Lemma $\ref{BM}$ and the operator $\mathcal L$ defined in \eqref{BR-Lax} are related via the following relation:
\begin{equation}
\label{rel-1}
\Ql^{1/2}\res\mathcal L^k = \res \overline L^k,\qquad k\geq 1,
\end{equation}
where the functions $u(x)$ and $w(x)$ in the above identity are identified by the Miura-type transformation
\begin{equation}
\label{Miura}
u(x) = p\frac{\Ql^{1/2}-\Ql^{-1/2}}{\qe\qp_x}w(x).
\end{equation}
\end{Th}
\begin{proof}
We start with considering the dressing operator $Q$ of the operator $\ol$ described in Lemma~\ref{BM}. From the definition of $Q$ we see that its residue $q_0$ satisfies the following equation:
\begin{equation}\label{temp10}
q_0 = {\rm e}^u\Ql^{-1}q_0.
\end{equation}
Taking derivatives of both sides of the above equation with respect to $x$, we arrive at
\begin{equation}\label{temp9}
\frac{\qp_xq_0}{q_0} = \frac{\qp_x}{1-\Ql^{-1}}u.
\end{equation}

Now let us perform a gauge transformation on $\overline L$ by the adjoint action of $q_0$ to obtain
\[
\overline{\mathscr L} = q_0^{-1}\overline L q_0 = \mathscr Q\Ql^{-1}\mathscr Q^{-1},\qquad
\mathscr Q = q_0^{-1}Q.
\]
By using the relation
\[
\ol-\frac 1p\log\ol = K
\]
together with \eqref{temp10} and \eqref{temp9}, we can check the validity of the following relation:
\begin{equation}\label{temp11}
\overline{\mathscr L}-\frac 1p \log\overline{\mathscr L} = \Ql^{-1}+\frac 1p \frac{\qe\qp_x}{1-\Ql^{-1}}u+\frac 1p \qe\qp_x,
\end{equation}
here the logarithm of $\overline{\mathscr L}$ is defined by $\mathscr Q$ as follows:
\[
\log\overline{\mathscr L} = -\qe\qp_x+\qe\mathscr Q_x\mathscr Q^{-1}.
\]
Denote $\overline{\mathscr L} = \Ql^{-1}+\sum_{k\geq 0} f_k\Ql^k$, then it follows from the identity
\[
\biggl[\overline{\mathscr L},\Ql^{-1}+\frac 1p \frac{\qe\qp_x}{1-\Ql^{-1}}u+\frac 1p \qe\qp_x\biggr] = 0
\]
that we can rewrite \eqref{temp11} as the following recursion relations for the coefficients $f_k$:
\begin{equation}\label{temp12}
f_k\big(\Ql^k-1\big)\Ql^{1/2}w+\big(1-\Ql^{-1}\big)f_{k+1} = \frac 1p \qe\qp_x f_k,\qquad
k\geq 0,\qquad
f_0 = \Ql^{1/2}w.
\end{equation}
Here we have already identified $u(x)$ and $w(x)$ via the relation \eqref{Miura}. The dispersionless limits of $f_k$ can be computed by using Lemma \ref{BM} and we obtain that $\lim_{\qe\to 0}f_k = \qd_k$, where $\qd_k$ are polynomials given by
\[
\qd_{k+1} = \frac 1p\qd_k-\int_0^wk\qd_k\,{\rm d}w,\qquad k\geq 0,\qquad \qd_0 = w.
\]

Recall that for a difference operator $\sum_{k}g_k\Ql^k$, its adjoint is defined to be $\big(\sum_{k}g_k\Ql^k\big)^* = \sum_k\big(\Ql^{-k}g_k\big)\Ql^{-k}$. Therefore if we denote $\overline{\mathscr L}^* = \Ql+\sum_{k\geq 0} f_k^*\Ql^{-k}$, we easily obtain from \eqref{temp12} the following recursion relations:
\begin{equation}\label{temp13}
f_k^*\big(1-\Ql^{-k}\big)\Ql^{1/2}w+(\Ql-1)f_{k+1} = \frac 1p \qe\qp_x f_k^*,\qquad
k\geq 0,
\qquad f_0^* = \Ql^{1/2}w.
\end{equation}
Finally by comparing \eqref{temp13} with \eqref{BV}, we arrive at
\[
\Ql^{1/2}\mathcal L\Ql^{-1/2} = \overline{\mathscr L}^*,
\]
which implies \eqref{rel-1}. The theorem is proved.
\end{proof}

The following corollary is straightforward and gives the exact correspondence between the flows of the LFV hierarchy and the ILW hierarchy:
\begin{Cor}\label{temp16}
The flows $\diff{w}{s_{1,n}}$ of the ILW hierarchy defined in \eqref{def1-ILW} coincide with the flows $\diff{u}{t_{2,n}}$ of the LFV hierarchy given in \eqref{def-LFVH} after identifying $w(x)$ and $u(x)$ by the Miura-type transformation \eqref{Miura}.
\end{Cor}
\begin{proof}From the definition \eqref{def1-ILW} of the ILW flows it follows that
\[\qe\diff{w}{s_{1,n}} = \frac 1p \qe\qp_x\res\mathcal L^n.\]
By using the Miura-type transformation \eqref{Miura} we arrive at
\[\qe\diff{u}{s_{1,n}} = \big(\Ql^{1/2}-\Ql^{-1/2}\big) \res\mathcal L^n.\]
On the other hand, from the definition \eqref{def-LFVH} of the LFV flows we know that
\[\qe\diff{u}{t_{2,n}} = \big(1-\Ql^{-1}\big) \res\overline L^k,\]
hence we complete the proof of the corollary by using the identity \eqref{rel-1}.
\end{proof}

From Corollary \ref{temp16} we conclude that after the Miura-type transformation \eqref{Miura}, the flows $\diff{u}{t_{1,n}}$ of the LFV hierarchy are transformed to symmetries of the ILW hierarchy, which has a Lax pair description given by the following constructions. The proofs of what follows are similar to the ones for the LFV hierarchy and we omit the details.
\begin{Lem}\label{zh-1-9}
There exists a difference operator
\begin{equation*}
\overline{\mathcal L} = d_{-1}\Ql^{-1}+\sum_{k\geq 0}d_k\Ql^k = \mathcal Q\Ql^{-1}\mathcal Q^{-1},\qquad d_k\in\mathcal R(w),
\end{equation*}
such that the following relation holds true
\begin{equation}
\label{BR-Lax2}
\frac 1p \log \overline{\mathcal L} = \mathcal K,
\end{equation}
here $\mathcal K$ is defined in \eqref{BR-Lax} and the logarithm of $\overline{\mathcal L}$ is defined by
\[\log\overline{\mathcal L} = -\qe\qp_x+\qe\mathcal Q_x\mathcal Q^{-1}.	\]
\end{Lem}
\begin{Def}
We define the following symmetries for the ILW hierarchy:
\begin{equation}
\label{def2-ILW}
\qe\diff{\mathcal K}{s_{2,n}} = -\big[\,\overline{\mathcal L}^n_-,\mathcal K\big],
\end{equation}
where $\overline{\mathcal L}$ is introduced in Lemma \ref{zh-1-9}.
\end{Def}
\begin{Ex}
The first flow defined in \eqref{def2-ILW} reads
\[
\qe\diff{w}{s_{2,1}} = (\Ql-1)\exp\bigg(p\frac{1-\Ql^{-1}}{\qe\qp_x}w\bigg).
\]
\end{Ex}

In a similar way as we prove Theorem \ref{temp15}, we can prove the following theorem.
\begin{Th}
The operator $ L$ defined in Lemma $\ref{BI}$ and the operator $\overline{\mathcal L}$ defined in \eqref{BR-Lax2} satisfy the following identity:
\begin{equation*}
\Ql^{1/2}\res\overline{\mathcal L}^k = \res L^k,\qquad k\geq 1,
\end{equation*}
where the functions $u(x)$ and $w(x)$ in the above identity are related by \eqref{Miura}.
\end{Th}
\begin{Cor}
The flows $\diff{w}{s_{2,n}}$ defined in \eqref{def2-ILW} coincide with the flows $\diff{u}{t_{1,n}}$ of the LFV hierarchy \eqref{def-LFVH} after identifying $w(x)$ and $u(x)$ by \eqref{Miura}. In particular, we have
\[
\bigg[\diff{}{s_{1,n}},\diff{}{s_{2,m}}\bigg] = \bigg[\diff{}{s_{2,n}},\diff{}{s_{2,m}}\bigg] = 0,\qquad n,m\geq 1.
\]
\end{Cor}

\section{Concluding remarks}\label{sec5}
In this paper, we consider the limiting procedure from the special cubic Hodge integrals to the linear Hodge integrals in view of the theory of integrable hierarchies. By taking a certain limit of the FV hierarchy, we obtain an integrable hierarchy which, together with the ILW hierarchy, forms a reduction of the 2D Toda hierarchy. We call the resulting hierarchy the LFV hierarchy, which is a Hamiltonian tau-symmetric integrable hierarchy with hydrodynamic limit.

This limiting procedure is quite natural in geometric setting, for example in the Gromov--Witten theory or the Hurwitz theory. Our result can be viewed as an integrable hierarchy theoretical interpretation of the Bouchard--Mari\~{n}o conjecture \cite{bouchard2007hurwitz} which is proved in \cite{mulase2010polynomial}. In \cite{bouchard2007hurwitz}, it is conjectured that the generating function of linear Hodge integrals can be computed in the scheme of Eynard and Orantin topological recursion associated with the spectral curve
\[C = \big\{x = y{\rm e}^{-y}\big\},\]
which is related to the symbol of the constraint \eqref{BN}. Their conjecture is based on a limiting procedure of the Mari\~{n}o--Vafa formula which relates the open amplitude of the A-model topological string on $\mathbb C^3$ with a framed brane on one leg of the toric diagram to the cubic Hodge integrals \cite{liu2003proof,liu2003mari,marino2001}. The Mari\~{n}o--Vafa formula can be used to derive a special case of Hodge-FVH correspondence \cite{takasaki2019cubic}, and the results of the present paper explain the above limiting procedure in view of the theory of integrable hierarchies. We thank the anonymous referee for pointing out this relation to us.

\subsection*{Acknowledgements}

This work is supported by NSFC no.~12171268, no.~11725104 and no.~11771238. We thank the anonymous referees for helpful comments and suggestions to improve the presentation of the paper.

\pdfbookmark[1]{References}{ref}
\LastPageEnding


\begin{thebibliography}{99}
\footnotesize\itemsep=0pt

\bibitem{bouchard2007hurwitz}
Bouchard V., Mari\~no M., Hurwitz numbers, matrix models and enumerative
 geometry, in From {H}odge Theory to Integrability and {TQFT} tt*-Geometry,
 \textit{Proc. Sympos. Pure Math.}, Vol.~78, \href{https://doi.org/10.1090/pspum/078/2483754}{Amer. Math. Soc.}, Providence, RI,
 2008, 263--283, \href{https://arxiv.org/abs/0709.1458}{arXiv:0709.1458}.

\bibitem{brini2012local}
Brini A., The local {G}romov--{W}itten theory of {$\mathbb{C}\mathbb{P}^1$} and
 integrable hierarchies, \href{https://doi.org/10.1007/s00220-012-1517-9}{\textit{Comm. Math. Phys.}} \textbf{313} (2012),
 571--605, \href{https://arxiv.org/abs/1002.0582}{arXiv:1002.0582}.

\bibitem{buryak2015dubrovin}
Buryak A., Dubrovin--{Z}hang hierarchy for the {H}odge integrals,
 \href{https://doi.org/10.4310/CNTP.2015.v9.n2.a1}{\textit{Commun. Number Theory Phys.}} \textbf{9} (2015), 239--272,
 \href{https://arxiv.org/abs/1308.5716}{arXiv:1308.5716}.

\bibitem{buryak2018simple}
Buryak A., Rossi P., Simple {L}ax description of the {ILW} hierarchy,
 \href{https://doi.org/10.3842/SIGMA.2018.120}{\textit{SIGMA}} \textbf{14} (2018), 120, 7~pages, \href{https://arxiv.org/abs/1809.00271}{arXiv:1809.00271}.

\bibitem{carlet2004extended}
Carlet G., Dubrovin B., Zhang Y., The extended {T}oda hierarchy, \href{https://doi.org/10.17323/1609-4514-2004-4-2-313-332}{\textit{Mosc.
 Math.~J.}} \textbf{4} (2004), 313--332, \href{https://arxiv.org/abs/nlin.SI/0306060}{arXiv:nlin.SI/0306060}.

\bibitem{dubrovin2016hodge}
Dubrovin B., Liu S.-Q., Yang D., Zhang Y., Hodge integrals and tau-symmetric
 integrable hierarchies of {H}amiltonian evolutionary {PDE}s, \href{https://doi.org/10.1016/j.aim.2016.01.018}{\textit{Adv.
 Math.}} \textbf{293} (2016), 382--435, \href{https://arxiv.org/abs/1409.4616}{arXiv:1409.4616}.

\bibitem{getzler2001toda}
Getzler E., The {T}oda conjecture, in Symplectic Geometry and Mirror Symmetry
 ({S}eoul, 2000), \href{https://doi.org/10.1142/9789812799821_0003}{World Sci. Publ.}, River Edge, NJ, 2001, 51--79,
 \href{https://arxiv.org/abs/math.AG/0108108}{arXiv:math.AG/0108108}.

\bibitem{gosper1978decision}
Gosper Jr. R.W., Decision procedure for indefinite hypergeometric summation,
 \href{https://doi.org/10.1073/pnas.75.1.40}{\textit{Proc. Nat. Acad. Sci. USA}} \textbf{75} (1978), 40--42.

\bibitem{liu2003proof}
Liu C.-C.M., Liu K., Zhou J., A proof of a conjecture of {M}ari\~no--{V}afa on
 {H}odge integrals, \href{https://doi.org/10.4310/jdg/1090511689}{\textit{J.~Differential Geom.}} \textbf{65} (2003),
 289--340, \href{https://arxiv.org/abs/math.AG/0306434}{arXiv:math.AG/0306434}.

\bibitem{liu2003mari}
Liu C.-C.M., Liu K., Zhou J., Mari\~no--{V}afa formula and {H}odge integral
 identities, \href{https://doi.org/10.1090/S1056-3911-05-00419-4}{\textit{J.~Algebraic Geom.}} \textbf{15} (2006), 379--398,
 \href{https://arxiv.org/abs/math.AG/0308015}{arXiv:math.AG/0308015}.

\bibitem{liu2021hodge}
Liu S.-Q., Yang D., Zhang Y., Zhou C., The {H}odge-{FVH} correspondence,
 \href{https://doi.org/10.1515/crelle-2020-0051}{\textit{J.~Reine Angew. Math.}} \textbf{775} (2021), 259--300,
 \href{https://arxiv.org/abs/1906.06860}{arXiv:1906.06860}.

\bibitem{liu2018fractional}
Liu S.-Q., Zhang Y., Zhou C., Fractional {V}olterra hierarchy, \href{https://doi.org/10.1007/s11005-017-1006-3}{\textit{Lett.
 Math. Phys.}} \textbf{108} (2018), 261--283, \href{https://arxiv.org/abs/1702.02840}{arXiv:1702.02840}.

\bibitem{marino2001}
Mari\~no M., Vafa C., Framed knots at large {$N$}, in Orbifolds in Mathematics
 and Physics ({M}adison, {WI}, 2001), \textit{Contemp. Math.}, Vol.~310, \href{https://doi.org/10.1090/conm/310/05404}{Amer.
 Math. Soc.}, Providence, RI, 2002, 185--204, \href{https://arxiv.org/abs/hep-th/0108064}{arXiv:hep-th/0108064}.

\bibitem{mulase2010polynomial}
Mulase M., Zhang N., Polynomial recursion formula for linear {H}odge integrals,
 \href{https://doi.org/10.4310/CNTP.2010.v4.n2.a1}{\textit{Commun. Number Theory Phys.}} \textbf{4} (2010), 267--293,
 \href{https://arxiv.org/abs/0908.2267}{arXiv:0908.2267}.

\bibitem{mumford1983towards}
Mumford D., Towards an enumerative geometry of the moduli space of curves, in
 Arithmetic and Geometry, {V}ol.~{II}, \textit{Progr. Math.}, Vol.~36,
 \href{https://doi.org/10.1007/978-1-4757-9286-7_12}{Birkh\"auser Boston}, Boston, MA, 1983, 271--328.

\bibitem{okounkov2006equivariant}
Okounkov A., Pandharipande R., The equivariant {G}romov--{W}itten theory of
 {${\bf P}^1$}, \href{https://doi.org/10.4007/annals.2006.163.561}{\textit{Ann. of Math.}} \textbf{163} (2006), 561--605,
 \href{https://arxiv.org/abs/math.AG/0207233}{arXiv:math.AG/0207233}.

\bibitem{takasaki2018toda}
Takasaki K., Toda hierarchies and their applications, \href{https://doi.org/10.1088/1751-8121/aabc14}{\textit{J.~Phys.~A: Math.
 Theor.}} \textbf{51} (2018), 203001, 35~pages, \href{https://arxiv.org/abs/1801.09924}{arXiv:1801.09924}.

\bibitem{takasaki2019cubic}
Takasaki K., Cubic {H}odge integrals and integrable hierarchies of {V}olterra
 type, in Integrability, Quantization, and Geometry. {I}. {I}ntegrable
 Systems, \textit{Proc. Sympos. Pure Math.}, Vol.~103, Amer. Math. Soc.,
 Providence, RI, 2021, 481--502, \href{https://arxiv.org/abs/1909.13095}{arXiv:1909.13095}.

\bibitem{toda1967vibration}
Toda M., Vibration of a chain with nonlinear interaction, \href{https://doi.org/10.1143/JPSJ.22.431}{\textit{J.~Phys. Soc.
 Japan}} \textbf{22} (1967), 431--436.

\bibitem{ueno1984toda}
Ueno K., Takasaki K., Toda lattice hierarchy, in Group Representations and
 Systems of Differential Equations ({T}okyo, 1982), \textit{Adv. Stud. Pure
 Math.}, Vol.~4, \href{https://doi.org/10.2969/aspm/00410001}{North-Holland}, Amsterdam, 1984, 1--95.

\end{thebibliography}
\end{document}